\documentclass{llncs}
\pdfoutput=1
\usepackage[english]{babel}
\usepackage[T1]{fontenc}
\usepackage[utf8]{inputenc}
\usepackage{enumitem}
\usepackage{graphicx}
\usepackage{amsmath}
\usepackage{amssymb}
\newcommand{\mEndProof}{\hfill$\square$} 
\title{ The Crossing Number of Seq-Shellable Drawings of Complete Graphs }
\author{ Petra Mutzel \and Lutz Oettershagen }
\institute{
	Department of Computer Science\\TU Dortmund University 
}
\date{\today} 
\begin{document}
\maketitle
\begin{abstract}
	The Harary-Hill conjecture states that for every $n>0$ the complete graph on $n$ vertices $K_n$, the minimum number of crossings over all its possible drawings equals
	\begin{align*}
	H(n) := \frac{1}{4}\Big\lfloor\frac{n}{2}\Big\rfloor\Big\lfloor\frac{n-1}{2}\Big\rfloor\Big\lfloor\frac{n-2}{2}\Big\rfloor\Big\lfloor\frac{n-3}{2}\Big\rfloor\text{.}
	\end{align*} 
	So far, the lower bound of the conjecture could only be verified for arbitrary drawings of $K_n$ with $n\leq 12$.
	In recent years, progress has been made in verifying the conjecture for certain classes of drawings,
	for example $2$-page-book, $x$-monotone, $x$-bounded, shellable and bishellable drawings. 
	Up to now, the class of bishellable drawings was the broadest class for which the Harary-Hill conjecture has been verified, as it contains all beforehand mentioned classes.
	In this work, we introduce the class of \emph{seq-shellable} drawings and verify the Harary-Hill conjecture for this new class.
	We show that bishellability implies seq-shellability and exhibit a non-bishellable but seq-shellable drawing of $K_{11}$, therefore the class of seq-shellable drawings strictly contains the class of bishellable drawings.
\end{abstract}
\section{Introduction}
Let $G=(V,E)$ be an undirected graph and $K_n$ the complete graph on $n>0$ vertices.
The crossing number $cr(G)$ of $G$ is the smallest number of edge crossings over all possible drawings of $G$. 
In a drawing $D$ every vertex $v \in V$ is represented by a point and every edge $uv\in E$ with $u,v \in V$ is represented by a simple curve connecting the corresponding points of $u$ and $v$.
The Harary-Hill conjecture states the following. % \cite{Guy1960}.
\begin{conjecture}[Harary-Hill \cite{Guy1960}]   %das Zitat finde ich hier wichtig, habe es wieder einkommentiert
	Let $K_n$ be the complete graph with $n$ vertices, then
	\begin{align*}
	cr(K_n) = H(n) &:= \frac{1}{4}\Big\lfloor\frac{n}{2}\Big\rfloor\Big\lfloor\frac{n-1}{2}\Big\rfloor\Big\lfloor\frac{n-2}{2}\Big\rfloor\Big\lfloor\frac{n-3}{2}\Big\rfloor \text{.}
	\end{align*} 
\end{conjecture}
There are construction methods for drawings of $K_n$ that lead to exactly $H(n)$ crossings,
for example the class of \emph{cylindrical} drawings first described by Harary and Hill \cite{hararyhill1963}.
For a cylindrical drawing, we put $\lfloor\frac{n}{2}\rfloor$ vertices on the top rim and the remaining $\lceil\frac{n}{2}\rceil$ vertices on the bottom rim of a cylinder.
Edges between vertices on the same rim (lid or bottom) are connected with straight lines on the lid or bottom. Two vertices on opposite rims are connected with an edge along the geodesic between the two vertices. The drawing of $K_6$ in figure \ref{fig:example} (a) is homeomorphic to a planarized cylindrical drawing of $K_6$.

However, there is no proof for the lower bound of the conjecture for arbitrary drawings of $K_n$ with $n>12$.
The cases for $n\leq 10$ are shown by Guy \cite{Guy1960} and for $n=11$ by Pan and Richter \cite{DBLP:journals/jgt/PanR07}. 
Guy \cite{Guy1960} argues that $cr(K_{2n+1})\geq H(2n+1)$ implies $cr(K_{2(n+1)})\geq H(2(n+1))$, hence $cr(K_{12})\geq H(12)$.
McQuillan et al. showed that $cr(K_{13})\geq 219$ \cite{DBLP:journals/jct/McQuillanPR15}. 
{\'A}brego et al. \cite{abrego2015all} improved the result to $cr(K_{13})\in\{223,225\}$. 

Beside these results for arbitrary drawings, there has been success in proving the Harary-Hill conjecture for different classes of drawings. 
So far, the conjecture has been verified for 2-page-book \cite{Abrego:2012:CNK:2261250.2261310}, $x$-monotone \cite{DBLP:journals/dcg/AbregoAF0S14,DBLP:journals/dcg/BalkoFK15,ABREGO2013411}, $x$-bounded  \cite{DBLP:journals/dcg/AbregoAF0S14}, shellable \cite{DBLP:journals/dcg/AbregoAF0S14} and bishellable drawings \cite{AbregoAFMMM0RV15}.
The class of bishellable drawings comprises all beforehand mentioned classes, and until now it was the largest class of drawings for which the Harary-Hill conjecture has been verified.
{\'A}brego et al. \cite{AbregoAFMMM0RV15} showed that the Harary-Hill conjecture holds for bishellable drawings using cumulated $k$-edges. 
\paragraph{Our contribution.}
In this work, we introduce the new class of \emph{seq-shellable} drawings and verify the Harary-Hill conjecture for this new class. 
We show that bishellability implies seq-shellability and exhibit a drawing of $K_{11}$ which is seq-shellable but not bishellable.
%Therefore, we establish that the set of seq-shellable drawings is a strict superset of the set of bishellable drawings.
Therefore, we establish that the class of seq-shellable drawings is strictly larger than the class of bishellable drawings.

The outline of this paper is as follows. In section \ref{sec:soa} we present the preliminaries, and in particular the background on $k$-edges, cumulated $k$-edges and their usage for verifying the Harary-Hill conjecture.
In section \ref{sec:seqsehll} we define \emph{simple sequences} and their usage for proving lower bounds on the number of invariant edges.
We present the definition of seq-shellability, verify the Harary-Hill conjecture for the new class and show its superiority towards the class of bishellable drawings.
Finally, in section \ref{sec:conclusions} we draw our conclusion and close with open questions.

\section{Preliminaries}\label{sec:soa}
Formally, a \emph{drawing} $D$ of a graph $G$ on the plane is an injection $\phi$ from the vertex set $V$ into the plane, and a mapping of the edge set $E$ into the set of simple curves, such that the curve corresponding to the edge $e = uv$ has endpoints $\phi(u)$ and $\phi(v)$, and contains no other vertices \cite{szekely2000successful}. % $\mathcal{C}$
We call an intersection point of the interior of two edges a crossing and a shared endpoint of two adjacent edges is not considered a crossing. 
The crossing number $cr(D)$ of a drawing $D$ equals the number of crossings in $D$ and the crossing number $cr(G)$ of a graph $G$ is the minimum crossing number over all its possible drawings.
We restrict our discussions to \emph{good} drawings of $K_n$, and call a drawing \emph{good} if $(1)$ any two of the curves have finitely many points in common, $(2)$ no two curves have a point in common in a tangential way, $(3)$ no three curves cross each other in the same point, $(4)$ any two edges cross at most once and $(5)$ no two adjacent edges cross. It is known that every drawing with a minimum number of crossings is good \cite{schaefer2013graph}.
In the discussion of a drawing $D$, we call the points also vertices, the curves edges and 
$V$ denotes the set of vertices (i.e. points), and $E$ denotes the edges (i.e curves) of $D$. 
If we subtract the drawing $D$ from the plane, a set of open discs remain. 
We call $\mathcal{F}(D) := \mathbb{R}^2 \setminus D$ the set of \emph{faces} of the drawing $D$.
If we remove a vertex $v$ and all its incident edges from $D$, we get the subdrawing $D-v$. 
Moreover, we might consider the drawing to be on the surface of the sphere $S^2$, which is equivalent to the drawing on the plane due to the homeomorphism between the plane and the sphere minus one point. 

In \cite{AbregoAFMMM0RV15} {\'A}brego et al. introduce bishellable drawings.
\begin{definition}[Bishellability \cite{AbregoAFMMM0RV15}] \label{def:bishell}
	For a non-negative integer $s$, a drawing $D$ of $K_n$ is $s$-bishellable if there exist sequences
	$a_0, a_1, \ldots, a_s$ and $b_s, b_{s-1}, \ldots,$ $ b_1 , b_0$, each sequence consisting of distinct vertices of $K_n$, so that with
	respect to a reference face $F$:
	\begin{enumerate}[leftmargin=10mm, label=(\roman*)]
		\item For each $i \in\{ 0, \ldots , s\}$, the vertex $a_i$ is incident to the face of \\$D - \{a_0, a_1 , \ldots , a_{i-1}\}$ that
		contains $F$,
		\item for each $i \in\{ 0, \ldots , s\}$, the vertex $b_i$ is incident to the face of \\$D - \{b_0, b_1 , \ldots , b_{i-1}\}$ that
		contains $F$, and
		\item for each $i \in\{ 0, \ldots , s\}$, the set $\{a_0, a_1 , \ldots a_i\} \cap \{b_{s-i}, b_{s-i-1} , \ldots , b_0\}$ \\is empty.
	\end{enumerate}
\end{definition}
The class of bishellable drawings contains all drawings that are $(\lfloor\frac{n}{2}\rfloor-2)$-bishellable. 
In order to show that if a drawing $D$ is $(\lfloor\frac{n}{2}\rfloor-2)$-bishellable, the Harary-Hill conjecture holds for $D$, {\'A}brego et al. use the notion of $k$-edges. 
The origins of $k$-edges lie in computational geometry and problems over $n$-point set, especially problems on halving lines and $k$-set \cite{abrego2012k}.
An early definition in the geometric setting goes back to Erd\H{o}s et al \cite{erdos1973dissection}.
Given a set $P$ of $n$ points in general position in the plane,
the authors add a directed edge $e=(p_i,p_j)$ between the two distinct points $p_i$ and $p_j$, and consider the continuation as line that separates the plane into the left and right half plane. There is a (possibly empty) point set $P_L\subseteq P$ on the left side of $e$, i.e. left half plane. Erd\H{o}s et al. assign $k:= \min(|P_L|, |P\setminus P_L|)$ to $e$.
Later, the name $k$-edge emerged and Lov{\'a}sz et al. \cite{lovasz2004convex} used $k$-edges for determining a lower bound on the crossing number of rectilinear graph drawings. Finally, 
{\'A}brego et al. \cite{Abrego:2012:CNK:2261250.2261310} extended the concept of $k$-edges from rectilinear to topological graph drawings. 

Every edge in a good drawing $D$ of $K_n$ is a $k$-edge with $k\in\{0,\ldots,\lfloor\frac{n}{2}\rfloor-1\}$. Let $D$ be on the surface of the sphere $S^2$, and $e=uv$ be an edge in $D$ and $F\in\mathcal{F}(D)$ be an arbitrary but fixed face; we call $F$ the \emph{reference face}.
Together with any vertex $w\in V\setminus\{u,v\}$, the edge $e$ forms a triangle $uvw$ and hence a closed curve that separates the surface of the sphere into two parts.
For an arbitrary but fixed orientation of $e$ one can distinguish between the left part and the right part of the separated surface. 
If $F$ lies in the left part of the surface, we say the triangle has orientation $+$ else it has orientation $-$. 
For $e$ there are $n-2$ possible triangles in total, of which $0\leq i\leq n-2$ triangles have orientation $+$ (or $-$) and $n-2-i$ triangles have orientation $-$ (or $+$ respectively).
%We define the \emph{$k$-value} of $e$ to be the minimum of $i$ and $n-2-i$.
%We say $e$ is an \emph{$i$-edge} with respect to the reference face $F$ precisely if its $k$-value equals $i$. 
We define $k:=\min(i,n-2-i)$ and say $e$ is an \emph{$k$-edge} with respect to the reference face $F$ and its \emph{$k$-value} equals $k$ with respect to $F$. 
{\'A}brego et al. \cite{Abrego:2012:CNK:2261250.2261310} show that the crossing number of a drawing is expressible in terms of the number of $k$-edges for $0\leq k\leq \lfloor\frac{n}{2}\rfloor-1$ with respect to the reference face. 
The following definition of the \emph{cumulated} number of $k$-edges is
helpful in determining the lower bound of the crossing number.
%
%TODO%%\todo[inline]{Eigentlich muesste man ja auch EkF(D) schreiben, vielleicht einmal so schreiben und dann sagen: wegen der besseren Uebersicht laesst man das F weg, aber wichtig: nicht vergessen, dass es von F ahaengt.}
%
\begin{definition}[Cumulated $k$-edges \cite{Abrego:2012:CNK:2261250.2261310}]
	Let $D$ be good drawing and $E_{k}(D)$ be the number of $k$-edges in $D$ with respect to a reference face $F\in \mathcal{F}(D)$ and for $k\in\{0,\ldots, \lfloor\frac{n}{2}\rfloor-1\}$.
	We call
	\vspace{-3mm}
	\begin{align*}
		E_{\leq\leq k}(D) := \sum_{i=0}^{k}(k+1-i)E_{i}(D) %\sum_{i=0}^{k}E_{\leq i}(D) =
	\end{align*}
	the cumulated number of $k$-edges with respect to $F$. 
\end{definition}
We also write \emph{cumulated $k$-edges} or \emph{cumulated $k$-value} instead of cumulated number of $k$-edges. 
Lower bounds on $E_{\leq\leq k}(D)$ for $0\leq k\leq \lfloor \frac{n}{2} \rfloor -2$ translate directly into a lower bound for $cr(D)$. 
\begin{lemma}\label{lemma:doublecumubound}\emph{\cite{Abrego:2012:CNK:2261250.2261310}}
	Let $D$ be a good drawing of $K_n$ and $F\in\mathcal{F}(D)$. 
	If $E_{\leq\leq k}(D) \geq 3{k+3 \choose 3}$ for all $0\leq k\leq \lfloor \frac{n}{2} \rfloor -2$ with respect to $F$, then $cr(D)\geq H(n)$.
	\mEndProof
\end{lemma}
%
%In order to count the number of $k$-edges in a drawing we consider edges incident to a vertex at $F$, because 
If a vertex $v$ is incident to the reference face, the edges incident to $v$ have a predetermined distribution of $k$-values.
\begin{lemma}\label{lemma:vatf}\emph{\cite{Abrego:2012:CNK:2261250.2261310}}
	Let $D$ be a good drawing of $K_n$, $F\in \mathcal{F}(D)$ and $v \in V$ be a vertex incident to $F$.
	With respect to $F$, vertex $v$ is incident to two $i$-edges for $0\leq i \leq \lfloor \frac{n}{2}\rfloor -2$.
	Furthermore, if we label the edges incident to $v$ counter clockwise with $e_0,\ldots,e_{n-2}$ such that $e_0$ and $e_{n-2}$ are incident to the face $F$, then $e_i$ is a $k$-edge with $k=\min(i,n-2-i)$ for $0\leq i \leq n-2$.
	\mEndProof
\end{lemma}
Examples for lemma \ref{lemma:vatf} are the vertices incident to $F$ in figure \ref{fig:example}.
We denote the cumulated $k$-values for edges incident to a vertex $v$ in a drawing $D$ with $E_{\leq\leq k}(D, v)$.
Due to lemma \ref{lemma:vatf} it follows that $E_{\leq\leq k}(D,v) = \sum_{i=0}^{k}(k+1-i)\cdot2 = 2{k+2 \choose 2}$.

Next, we introduce \emph{invariant} $k$-edges. Consider removing a vertex $v \in V$ from a good drawing $D$ of $K_n$, resulting in the subdrawing $D-v$.  By deleting $v$ and its incident edges every remaining edge loses one triangle, i.e. for an edge $uw\in E$ there are only $(n-3)$ triangles $uwx$ with $x\in V\setminus\{u,v\}$ (instead of the $(n-2)$ triangles in drawing $D$). 
The $k$-value of any edge $e\in E$ is defined as the minimum count of $+$ or $-$ oriented triangles that contain $e$. If the lost triangle had the same orientation as the minority of triangles, the $k$-value of $e$ is reduced by one else it stays the same.
Therefore, every $k$-edge in $D$ with respect to $F\in \mathcal{F}(D)$ is either a $k$-edge or a $(k-1)$-edge in the subdrawing $D-v$ with respect to $F'\in \mathcal{F}(D-v)$ and $F\subseteq F'$.
We call an edge $e$ \emph{invariant} if $e$ has the same $k$-value with respect to $F$ in $D$ as for $F'$ in $D'$. We denote the number of cumulated invariant $k$-edges between $D$ and $D'$ (with respect to $F$ and $F'$ respectively) with $I_{\leq k}(D, D')$, i.e. $I_{\leq k}(D, D')$ equals the sum of the number of invariant $i$-edges for $0\leq i \leq k$.
\begin{figure}
	\centering
	\includegraphics[width=0.93\linewidth]{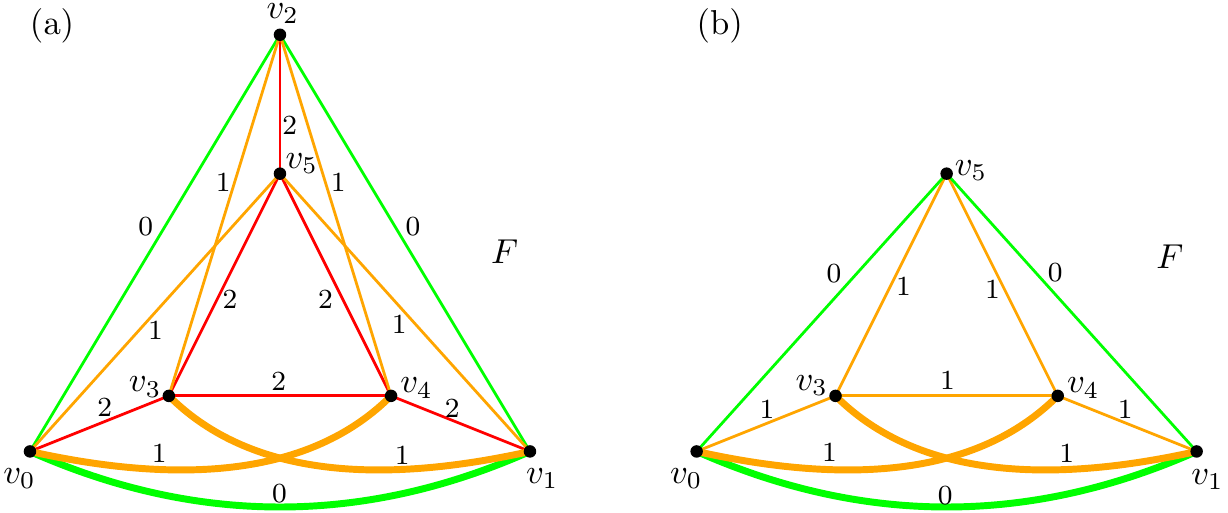} %images/example_02
	\caption{Example: (a) shows a crossing optimal drawing $D$ of $K_6$ with the $k$-values at the edges. (b) shows the subdrawing $D-v_2$ and its $k$-values. The fat highlighted edges $v_0v_1$, $v_0v_4$ and $v_1v_3$ are invariant and keep their $k$-values. The reference face is the outer face $F$.}
	%Also notice the distribution of the $k$-values of the edges incident to vertices that touch $F$ according to lemma \ref{lemma:vatf}.}
	\label{fig:example}
\end{figure}

For a good drawing $D$ of $K_n$, we are able to express the value of cumulated $k$-edges with respect to a reference face $F\in\mathcal{F}(D)$ recursively by adding up the cumulated $(k-1)$-value of a subdrawing $D-v$, the contribution of the edges incident to $v$ and the number of \emph{invariant edges} between $D$ and $D-v$. 
\begin{lemma}\label{lemma:recursive_double_cumu}\emph{\cite{AbregoAFMMM0RV15}}
	Let $D$ be a good drawing of $K_n$, $v\in V$ and $F\in \mathcal{F}(D)$. With respect to the reference face $F$, we have
	\begin{align*}
		E_{\leq\leq k}(D) &= E_{\leq\leq k-1}(D-v) + E_{\leq\leq k}(D, v) + I_{\leq k}(D,D-v) \textrm{.}
	\end{align*}
	\mEndProof
\end{lemma}
{\'A}brego et al. \cite{AbregoAFMMM0RV15} use an inductive proof over $k$ to show that for a bishellable drawing $D$ of $K_n$  $E_{\leq\leq k}(D)\geq 3{k+3 \choose 3}$ for all $k\in\{0,\ldots,\lfloor\frac{n}{2}\rfloor-2\}$.
Together with lemma \ref{lemma:doublecumubound} follows $cr(D)\geq H(n)$. 

%%%
Here, we also use lemma \ref{lemma:recursive_double_cumu} and show that for a seq-shellable drawing $D$ of $K_n$ the lower bounds on $E_{\leq\leq k}(D)$ hold for all $k\in\{0,\ldots,\lfloor\frac{n}{2}\rfloor-2\}$.
But in contrast to \cite{AbregoAFMMM0RV15}, we use a more general and at the same time easy to follow approach to guarantee lower bounds on the number of invariant edges $I_{\leq k}(D,D-v)$ for $0\leq k \leq \lfloor\frac{n}{2}\rfloor-2$.
\section{Seq-Shellability}\label{sec:seqsehll}
Before we proceed with the definition of seq-shellability, we introduce simple sequences.
%We use simple sequences to guarantee a lower bound of the number of invariant edges in the recursive formulation of the cumulated $k$-value.

\subsection{Simple sequences}
We use simple sequences to guarantee a lower bound of the number of invariant edges in the recursive formulation of the cumulated $k$-value.

%%%\todo[inline]{Das ist einerseits auch gut, andererseits wieder problematisch; evtl.: wir def... in 4 implicit benutzt. Im Gegensatz dazu: bei uns flexibler....?}
%
\begin{definition}[Simple sequence] 
	Let $D$ be a good drawing of $K_n$, $F\in\mathcal{F}(D)$ and $v\in V$ with $v$ incident to $F$. 
	Furthermore, let $S_v = (u_0,\ldots,u_k)$ with $u_i\in V\setminus\{v\}$ be a sequence of distinct vertices.
	If $u_0$ is incident to $F$ and vertex $u_i$ is incident to a face containing $F$ in subdrawing $D-\{u_0,\ldots, u_{i-1}\}$ for all $1\leq i \leq k$, then we call $S_v$ simple sequence of $v$.
\end{definition}
%An alike but very restricted idea is
%%Simple sequences are used 
%implicitly used in the proof of the Harary-Hill conjecture for bishellable drawings in \cite{AbregoAFMMM0RV15}. Here we overcome these substantial restrictions and 
%We use simple sequences explicitly in a flexible manner that ultimately leads to the larger class of drawings for which we are able to verify the Harary-Hill conjecture.
%
%However, 
Before we continue with a result for lower bounds on the number of invariant edges using simple sequences, we need the following lemma.
\begin{lemma}\label{lemma:a0b0_at_F}
	Let $D$ be a good drawing of $K_n$, $F\in \mathcal{F}(D)$ and $u,v\in V$ with $u$ and $v$ incident to $F$.
	The edge $uv$ touches $F$ either over its full length or not at all (except its endpoints).
\end{lemma}
\begin{proof}
	Assume that $D$ a is good drawing of $K_n$ in which the edge $uv$ touches $F$ only partly. 
	We can exclude the case that an edge cuts a part out of $uv$ by crossing it more than once due to the goodness of the drawing (see figure \ref{fig:a0b0_at_F} (a)).
	The case that an edge crosses the whole face $F$ and separates it into two faces is also impossible, because this would contradict that both $u$ and $v$ are incident to $F$.
	Therefore, a vertex $x$ has to be on the same side of $uv$ as $F$ and a vertex $y$ on the other side such that the edge $xy$ crosses $uv$.
	But the edge $xu$ cannot cross any edge $uz$ with $z\in V\setminus\{u\}$ as this would contradict the goodness of $D$
	and $xu$ cannot leave the superface of $x$ without {separating} $v$ from $F$ (see figure \ref{fig:a0b0_at_F} (b) and (c)).
	We have the symmetric case for $v$. Consequently, $uv$ cannot touch $F$ beside its endpoints $u$ and $v$ (see figure \ref{fig:a0b0_at_F} (d)), a contradiction to the assumption.
	\mEndProof
\end{proof}
\begin{figure}
	\centering
	\includegraphics[width=1\linewidth]{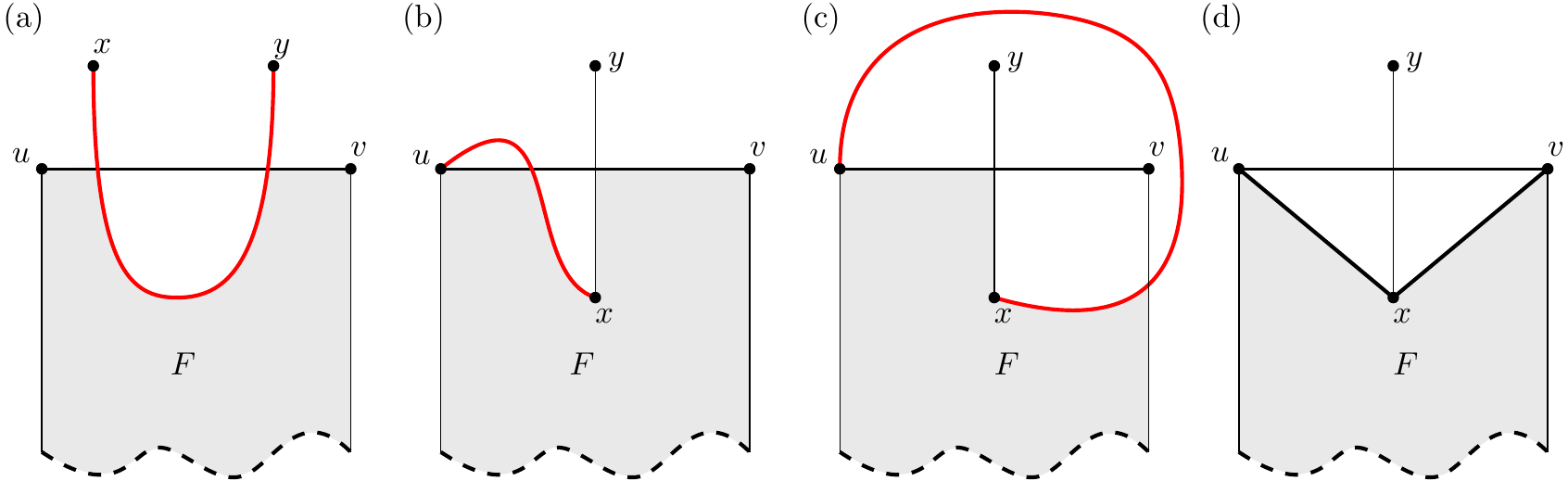} %images/a0b0_at_F_02_col_04
	\caption[Possible crossings of edge $uv$]{
		(a) Due to the goodness of $D$ an edge cannot \emph{cut} a part out of the edge $uv$.
		(b) The edges $uv$ and $ux$ cross, both have vertex $u$ as endpoint thus the drawing is not good.
		(c) The drawing is good but vertex $v$ is not incident to the face $F.$ 
		(d) The edge $uv$ is crossed, the drawing is good and both  vertices $u$ and $v$ is are incident to $F$, however $uv$ is not incident to $F$.}
	\label{fig:a0b0_at_F}
\end{figure}
\begin{corollary}\label{corollary:a0b0_j_edge}
	Let $D$ be a good drawing of $K_n$, $F\in \mathcal{F}(D)$ and $u,v\in V$ with both $u$ and $v$ incident to $F$.
	If and only if $uv$ is a $j$-edge, there are exactly $j$ or $n-2-j$ vertices on the same side of $uv$ as the reference face $F$.
	\mEndProof
\end{corollary}
The following lemma provides a lower bound for the number of invariant edges in the case that $F$ is incident to at least two vertices and we remove one of them.
\begin{lemma}\label{lemma:inv_vwatf}%\emph{(cf. \cite{AbregoAFMMM0RV15})} %PM: wuerde ich nicht machen
%	Let $D$ be a good drawing of $K_n$, $F\in\mathcal{F}(D)$ and $v,w \in V$ with $v$ and $w$ incident to $F$. If we remove $v$ from $D$, then $w$ is incident to at least $k+1$ invariant edges for all $k\in\{0,\ldots,\lfloor \frac{n}{2}\rfloor-2\}$.
	Let $D$ be a good drawing of $K_n$, $F\in\mathcal{F}(D)$ and $v,w \in V$ with $v$ and $w$ incident to $F$. If we remove $v$ from $D$, then $w$ is incident to at least $\lfloor \frac{n}{2}\rfloor-1$ invariant edges.
\end{lemma}
\begin{proof}
	We label the edges incident to $w$ counter clockwise with $e_0,\ldots,e_{n-2}$ such that $e_0$ and $e_{n-2}$ are incident to the face $F$, 
	and we label the vertex at the other end of $e_i$ with $u_i$. 
	Furthermore, we orient all edges incident to $w$ as outgoing edges.
	Due to lemma \ref{lemma:vatf} we know that $w$ has two $i$-edges for $0\leq i \leq\lfloor \frac{n}{2}\rfloor -2$. Edge $e_i$ obtains its $i$-value from the minimum of say $+$ oriented triangles and edge $e_{n-2-i}$ obtains its $i$-value from the minimum $-$ oriented triangles (or vice versa).
	Assume that $vw$ is incident to $F$, i.e. $vw$ is a $0$-edge and all triangles $vwu$ for $u\in V\setminus\{v,w\}$ have the same orientation. 
	Consequently, all $e_i$ or all $e_{n-2-i}$ for $0\leq i \leq \lfloor\frac{n}{2}\rfloor-2$ are invariant.
	In the case that $vw$ is not incident to $F$ and is a $j$-edge, there are $j$ triangles $vwu_h$ with $u_h\in V\setminus\{v,w\}$, $0\leq h\leq j-1$ or $n-1-j \leq h\leq n-2$ 
	and $u_h$ is on the same side of $vw$ as $F$ (corollary \ref{corollary:a0b0_j_edge}). 
	This means, each triangle $wu_hv$ is part of the majority of orientations for the $k$-value of edge $wu_h$, therefore removing $v$ does not change the $k$-value and 	
	there are $j$ additional invariant edges incident to $w$ if we remove $v$.
	\mEndProof
\end{proof}
The following lemma provides a lower bound for the number of cumulated invariant $k$-edges if we remove a vertex that has a simple sequence.
%The following lemma provides a lower bound for the number of invariant $(\leq k)$-edges if we remove a vertex that has a simple sequence.
%
\begin{lemma}\label{lemma:inv_vwatf_seq}
	Let $D$ be a good drawing of $K_n$, $F\in\mathcal{F}(D)$ and $v\in V$ with $v$ incident to $F$. 
	If $v$ has a simple sequence $S_v = (u_0,\ldots,u_k)$, then
%	\vspace{-1mm}
	\begin{align*}
		I_{\leq k}(D,D-v)\geq {k+2\choose 2}
	\end{align*}
	with respect to $F$ and for all $k\in \{0,\ldots,\lfloor\frac{n}{2}\rfloor-2\}$.
\end{lemma}
\begin{proof}
	Let $k\in \{0,\ldots,\lfloor\frac{n}{2}\rfloor-2\}$.
%	Vertices $v$ and $u_0$ are incident to $F$ and with lemma \ref{lemma:inv_vwatf} w
	We know that $u_0$ has at least $k+1\leq \lfloor \frac{n}{2}\rfloor-1$ invariant edges with respect to $F$ and removing $v$.
	After removing vertex $u_0$ from drawing $D$, vertices $v$ and $u_1$ are incident to $F$. 
	Since $k\leq\lfloor\frac{n}{2}\rfloor-2 \leq \lfloor\frac{n-1}{2}\rfloor-1$ and 
	$u_0$ has an edge to $u_1$ in drawing $D$, vertex $u_1$ has at least $k$ invariant edges with respect to $F$ and removing $v$ in drawing $D$.
	In general, after removing vertices $u_0,\ldots,u_{i-1}$ from drawing $D$, vertices $v$ and $u_i$ are incident to $F$.
	For $u\in\{u_0,\ldots,u_{i-1}\}$ the edge $uu_i$ in drawing $D$ may be invariant or non-invariant, and we have $k+1-i\leq\lfloor\frac{n}{2}\rfloor-1-i \leq \lfloor\frac{n-i}{2}\rfloor-1$.
	Therefore, $u_i$ has at least $k-i+1$ invariant edges in drawing $D$ with respect to $F$ and removing $v$. Summing up leads to %Summing up all invariant edges leads to
	\begin{align*}
		I_{\leq k}(D,D-v)\geq \sum_{i=0}^{k}(k+1-i)={k+2\choose 2}\text{.}
	\end{align*}
	\mEndProof
\end{proof}
\subsection{Seq-shellable drawings}
With help of simple sequences we define $k$-seq-shellability. For a sequence of distinct vertices $a_0,\ldots, a_k$ we assign to each vertex $a_i$ with $0\leq i \leq k\leq n-2$ a simple sequence $S_i$, under the condition  that $S_i$ does not contain any of the vertices $a_0,\ldots,a_{i-1}$. %, and such
%that $a_0,\ldots, a_k$ is a simple sequence for the first vertex of $S_0$.
%
\begin{definition}[Seq-Shellability]
	Let $D$ be a good drawing of $K_n$. We call $D$ \emph{$k$-seq-shellable} for $k\geq 0$ if there exists a face $F\in \mathcal{F}(D)$ and a sequence of distinct vertices $a_0,\ldots,a_k$ such that $a_0$ is incident to $F$ and
	\begin{enumerate}
		\item for each $i\in \{1,\ldots,k\}$, vertex $a_i$ is incident to the face containing $F$ in drawing $D-\{a_0,\ldots,a_{i-1}\}$ and
		\item for each $i\in \{0,\ldots,k\}$, vertex $a_i$ has a simple sequence $S_i=(u_0,\ldots,u_{k-i})$ with $u_j\in V\setminus \{a_0,\ldots,a_i\}$ for $0\leq j \leq k-i $ in drawing $D-\{a_0,\ldots,a_{i-1}\}$. %(or $D$ if $i=0$).
	\end{enumerate}
\end{definition}
Notice that if $D$ is $k$-seq-shellable for $k>0$, then the subdrawing $D-a_0$ is $(k-1)$-seq-shellable. Moreover, if $D$ is $k$-seq-shellable, 
%for $k \in \{0,\ldots,\lfloor\frac{n}{2}\rfloor-2\}$, 
then $D$ is also $j$-seq-shellable for $0\leq j\leq k$.
\begin{lemma}\label{lemma:seqsehll_01}
	If $D$ is a good drawing of $K_n$ and $D$ is $k$-seq-shellable %with $k\geq 0$,
	with $k \in \{0,\ldots,\lfloor\frac{n}{2}\rfloor-2\}$, 
	then $E_{\leq\leq k}(D)\geq 3{k+3\choose 3}$.
\end{lemma}
\begin{proof}
	We proceed with induction over $k$.
%	\\\textbf{Basis:}  %PM: nicht notwendig
	For $k=0$ the reference face is incident to at least three $0$-edges and it follows that
	\begin{align*}
	E_{\leq\leq 0}(D)\geq 3 = 3{0+3\choose 3}\text{.}
	\end{align*}
%	\\\textbf{Induction step:} %PM: nicht notwendig
	For the induction step, let $D$ be $k$-seq-shellable with $a_0,\ldots,a_k$ and the sequences $S_0,\ldots,S_k$.
	Consider the drawing $D-a_0$ which is $(k-1)$-seq-shellable for $a_1,\ldots,a_k$ and $S_1,\ldots,S_k$. Since
	$k - 1 \leq (\lfloor\frac{n}{2}\rfloor-2)-1 \leq (\lfloor\frac{n-1}{2}\rfloor-2)$, we assume %the induction implies that
	\begin{align*}
	E_{\leq\leq k-1}(D-a_0)\geq 3{k+2\choose 3}\text{.}
	\end{align*}
	We use the recursive formulation introduced in lemma \ref{lemma:recursive_double_cumu}, i.e.
	\begin{align*}
		E_{\leq\leq k}(D) &= E_{\leq\leq k-1}(D-a_0) + E_{\leq\leq k}(D, a_0) + I_{\leq k}(D,D-a_0) \textrm{.}
	\end{align*}
	Because $a_0$ is incident to $F$, we have $E_{\leq\leq k}(D, a_0)= 2{k+2 \choose 2}$, 
%	(see corollary \ref{corollary:vatf_cumus}), 
	and with the simple sequence $S_0$ of $a_0$ follows $I_{\leq k}(D,D-a_0)\geq {k+2 \choose 2}$ (see lemma \ref{lemma:inv_vwatf_seq}). Together with the induction hypothesis, we have 
	\begin{align*}
		E_{\leq\leq k}(D) &\geq 3{k+2\choose 3} +2{k+2\choose 2}+{k+2\choose 2}= 3{k+3\choose 3} \textrm{.}
	\end{align*}
	\mEndProof
\end{proof}
Using lemmas \ref{lemma:doublecumubound} and \ref{lemma:seqsehll_01}, we are able to verify the Harary-Hill conjecture for seq-shellable drawings.
\begin{theorem}
	If $D$ is a good drawing of $K_n$ and $D$ is $(\lfloor\frac{n}{2}\rfloor-2)$-seq-shellable, then $cr(D)\geq H(n)$.
\end{theorem}
\begin{proof}
	Let $D$ be a good drawing of $K_n$ and $(\lfloor\frac{n}{2}\rfloor-2)$-seq-shellable.
	Since $D$ is $(\lfloor\frac{n}{2}\rfloor-2)$-seq-shellable, it is also $k$-seq-shellable for $0\leq k \leq \lfloor\frac{n}{2}\rfloor-2$.
	We apply lemma \ref{lemma:seqsehll_01} and have $E_{\leq\leq k}(D)\geq 3{k+3\choose 3}$ for $0\leq k \leq \lfloor\frac{n}{2}\rfloor-2$
	and the result follows with lemma \ref{lemma:doublecumubound}.
	\mEndProof
\end{proof}
If a drawing $D$ of $K_n$ is $(\lfloor\frac{n}{2}\rfloor-2)$-seq-shellable, we omit the $(\lfloor\frac{n}{2}\rfloor-2)$ part and say $D$ is seq-shellable. 
The class of seq-shellable drawings contains all drawings that are $(\lfloor\frac{n}{2}\rfloor-2)$-seq-shellable.
%
%\subsection{Seq-shellability vs. bishellability}
%
%
%%
\begin{theorem}\label{theorem:bish_impl_ss}
%	The set of seq-shellable drawings is a strict superset of the set of bishellable drawings.
	The class of seq-shellable drawings strictly contains the class of bishellable drawings.
\end{theorem}
\begin{proof}
	First, we show that $k$-bishellability implies $k$-seq-shellability. 
	Let $D$ be a $k$-bishellable drawing of $K_n$ with the associated sequences $a_0,\ldots,a_k$ and $b_0,\ldots,b_k$.
	In order to show that $D$ is $k$-seq-shellable, we choose $a_0,\ldots,a_k$ as vertex sequence and $k$ simple sequences $S_i$ for $0\leq i \leq k$ such that $S_i=(b_0,\ldots,b_{k-i})$.
%	For each $0\leq i \leq k$ we assign to vertex $a_i$ simple sequence $S_i$ and we are done.
	We assign simple sequence $S_i$ to vertex $a_i$ for each $0\leq i \leq k$ and see that $D$ is indeed seq-shellable.
	Furthermore, drawing $H$ of $K_{11}$ in figure \ref{fig:k11_seqshell} is not bishellable but seq-shellable. 
	It is impossible to find sequences $a_0,\ldots, a_3$ and $b_0,\ldots,b_3$ in $H$ that fulfill the definition of bishellability.
	However, $H$ is seq-shellable for face $F$, vertex sequence $(v_0,v_2,v_3,v_4)$ and the simple sequences $S_0=(v_1,v_2,v_7,v_4)$, $S_1=(v_1,v_8,v_6)$, $S_2=(v_1,v_8)$ and $S_3=(v_1)$.
%	It follows that the class of seq-shellable drawings strictly contains the class of bishellable drawings.
%	Therefore, the result follows.
	\mEndProof
\end{proof}
\begin{figure}
	\centering
	\includegraphics[width=0.65\linewidth]{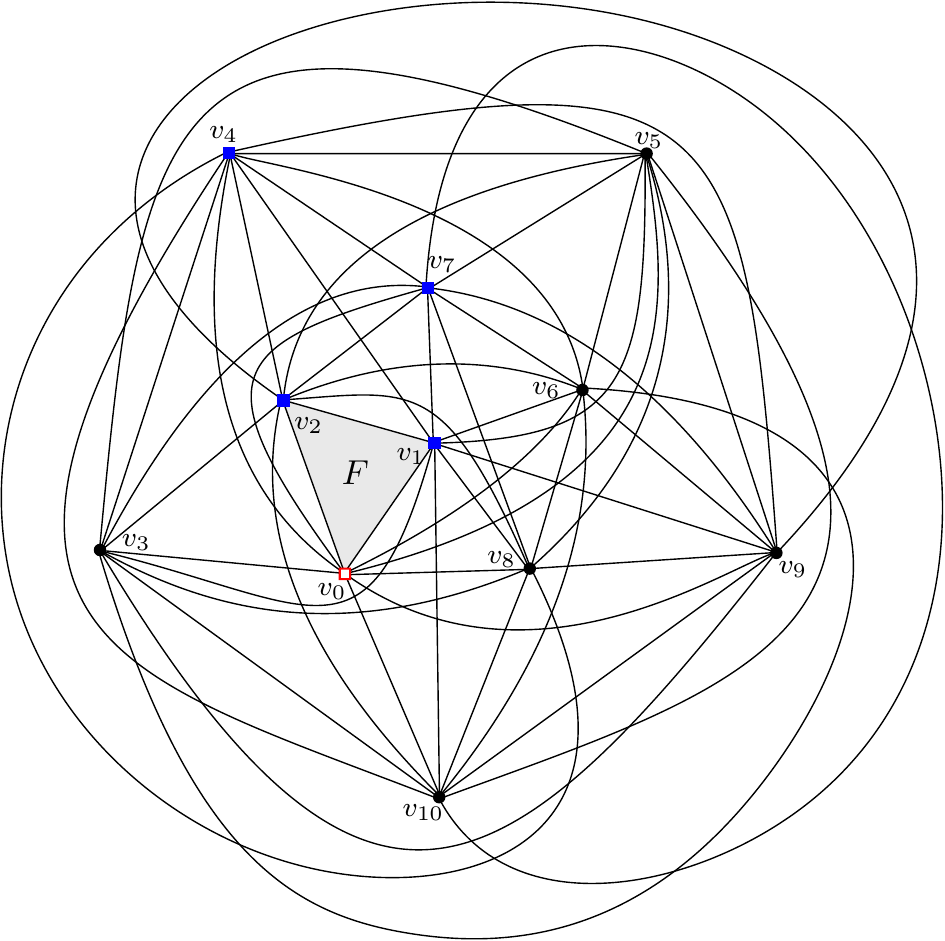} %images/k11_seq_shell_004
	\caption[Non-bishellable but seq-shellable drawing of $K_{11}$]{
		Drawing $H$ of $K_{11}$ which is not bishellable for any face, however it is seq-shellable for face $F$, vertex sequence $(v_0,v_2,v_3,v_4)$ and the simple sequences $S_0=(v_1,v_2,v_7,v_4)$, $S_1=(v_1,v_8,v_6)$, $S_2=(v_1,v_8)$ and $S_3=(v_1)$.
		Vertex $v_0$ and the vertices of $S_0$ are highlighted as unfilled and filled squares.}
	\label{fig:k11_seqshell}
	\vspace{1cm}
	\includegraphics[width=0.66\linewidth]{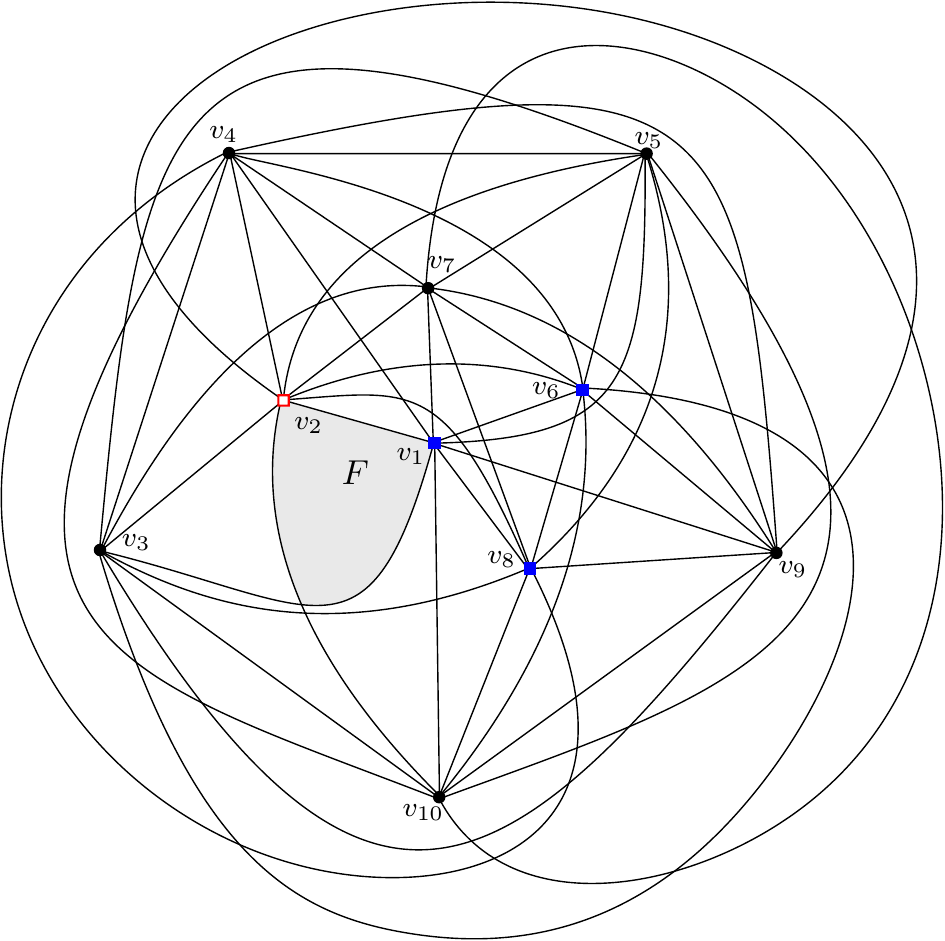} %images/k11_seq_shell_004_2
	\caption[Non-bishellable but seq-shellable drawing of $K_{11}$]{
		Subdrawing $H-v_0$ after removing vertex $v_0$ and its incident edges. The second vertex of the vertex sequence $v_2$ is incident to the face containing $F$ and has simple sequence $S_1$.
		Vertex $v_2$ and the vertices of $S_1$ are highlighted as unfilled and filled squares.}
	\label{fig:k11_seqshell_wo_v0}
\end{figure}

The distinctive difference between seq-shellability and bishellability is that the latter demands a symmetric structure in the sense that we can mutually exchange the sequences $a_0,\ldots,a_k$ and $b_0,\ldots,b_k$. Thus, the sequence $b_0,\ldots,b_{k-i}$ has to be the simple sequence of $a_i$ in the subdrawing $D-\{a_0,\ldots,a_{i-1}\}$ for all $0\leq i \leq k$ and vice versa, i.e. the sequence $a_0,\ldots,a_{k-i}$ has to be the simple sequence of $b_i$ in the subdrawing $D-\{b_0,\ldots,b_{i-1}\}$ for all $0\leq i \leq k$.
With seq-shellability we do not have this requirement. Here we have the vertex sequence $a_0,\ldots, a_k$ and each vertex $a_i$ with $0\leq i\leq k$ has its own (independent) simple sequence $S_i$.

Figure \ref{fig:ss_gadget_01} shows a gadget that visualizes the difference between bishellability and seq-shellability:
(a) shows a substructure with nine vertices that may occur in a drawing. We have the simple sequence $v_1,v_2,v_4$ for vertex $v_3$ in (b) and (c). Therefore, we can remove vertex $v_3$ and are able to guarantee the number of invariant edges.
After removing vertex $v_3$ in (d), there are simple sequences for vertex $v_1$ and $v_2$, thus the substructure is seq-shellable. 
However, it is impossible to apply the definition of bishellability.
We may use, for example, sequence $v_1,v_2,v_4$ as $a_0,\ldots,a_k$ sequence and we need a second sequence (the $b$ sequence) that satisfies the exclusion condition of the bishellability, i.e. for each $i \in\{ 0, \ldots , k\}$, the set $\{a_0, a_1 , \ldots a_i\} \cap \{b_{k-i}, b_{k-i-1} , \ldots , b_0\}$ has to be empty (see definition \ref{def:bishell}).
The first vertex of our second sequence (i.e. $b_0$) has to be $v_3$, because $b_0$ has to be incident to $F$. 
Now, for the second vertex we have to satisfy $\{a_0, a_1\} \cap \{b_1, b_0\}=\emptyset$, thus the second vertex has to be different from the first two vertices of the sequence $v_1,v_2,v_4$.
Because we only can choose between vertices $v_1$ and $v_2$, we cannot select a second vertex for our $b$ sequence. Thus, the structure is not bishellable.
%
%The edges $v_1v_5$ and $v_2v_6$ together are acting as \emph{wall} that prevents bishellability.
We can argue the same way for the other possible sequences in the gadget. 
%The edges of the vertices $v_1$, $v_2$ and $v_3$ are acting in pairs as \emph{walls} that prevent bishellability.

%Notice that a drawing in which every face is incident to only one vertex cannot be seq-shellable. %, 

\begin{figure}
	\centering
	\includegraphics[width=0.62\linewidth]{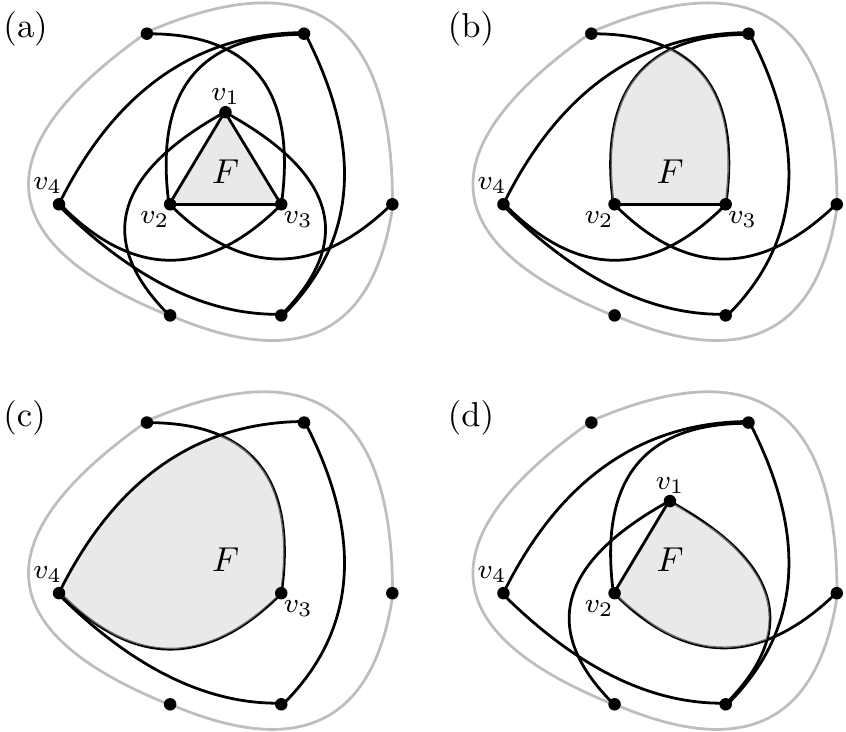} %images/seq_shell_gadget_02_nocol01
	\caption[A gadget that prohibits bishellability]{The gadget does not allow for a bishellability sequence, because only one of the two sequences $a_0,\ldots,a_k$ or $b_0,\ldots,b_k$ can be chosen due to condition three of the definition of bishellability. However, the gadget is seq-shellable.}
	\label{fig:ss_gadget_01}
\end{figure}
\section{Conclusion}\label{sec:conclusions}
In this work, we introduced the new class of seq-shellable drawings and verified the Harary-Hill conjecture for this class.
Seq-shellability is a generalization of bishellability, thus bishellability implies seq-shellability. In addition we exhibited a drawing of $K_{11}$ which is seq-shellable but not bishellable, hence seq-shellability is a proper extension of bishellability.
So far, we are not aware of an optimal seq-shellable but non-bishellable drawing and we close with the following open questions:
\begin{enumerate}
	\item %Do optimal seq-shellable but non-bishellable drawings exist and   
		Can we find a construction method to obtain optimal drawings of $K_n$ that are seq-shellable but not bishellable?
	\item Does there exists a non-bishellable but seq-shellable drawing of $K_n$ with $10 \leq n < 14$, such that after removing the first vertex of the simple sequence the drawing $D-a_0$ is still non-bishellable. We found a drawing of $K_{14}$ with this property. % (figure \ref{fig:k14_seqshell01}).
\end{enumerate}
%So far, we are not aware of a construction method to obtain optimal drawings of $K_n$ that are seq-shellable but not bishellable. We are not even aware of a single instance of a drawing with these properties.
%Furthermore, it is an open question if there exists a non-bishellable but seq-shellable drawing of $K_n$ with $10 \leq n < 14$, such that after removing the first vertex of the simple sequence the drawing $D-a_0$ is still non-bishellable. We found a drawing of $K_{14}$ with this property.
%
%\begin{figure}
%	\centering
%	\includegraphics[width=0.65\linewidth]{images/k14_seq_shell_04}
%	\caption[Seq-shellable drawing of $K_{14}$]{Seq-shellable drawing $D$ of $K_{14}$. If we remove any of the vertices the resulting subdrawing of $K_{13}$ is not bishellable.}
%	\label{fig:k14_seqshell01}
%\end{figure}

\bibliographystyle{splncs03}
\bibliography{bibliography/bibliography}

\end{document}